\documentclass[12pt]{amsart}
\usepackage{hyperref, amsmath,amssymb,amsfonts,amsbsy,mathtools,cite,setspace}
\usepackage{graphicx}
\usepackage{amssymb}
\usepackage{amsthm}
\usepackage{enumerate}
\usepackage{wasysym}
\usepackage{cite}
\usepackage{pgf,tikz}
\usepackage{mathrsfs}
\usetikzlibrary{arrows}
\usepackage{mathtools}

\usepackage{subcaption}
\usepackage{braket}
\usepackage{color}
\usepackage{float}

\def\ket#1{| #1 \rangle}
\def\bra#1{\langle #1 |}
\def\kb#1#2{|#1\rangle\!\langle #2 |}
\def\bk#1#2{\langle #1 |#2\rangle}

\def\be{\begin{eqnarray}}
\def\ee{\end{eqnarray}}
\def\bee{\begin{eqnarray*}}
\def\eee{\end{eqnarray*}}

\newtheorem{definition}{Definition}
\newtheorem{proposition}{Proposition}
\newtheorem{theorem}{Theorem}
\newtheorem{example}{Example}

\newtheorem{remark}{Remark}
\newtheorem{corollary}{Corollary}

\newcommand{\C}{{\mathbb C}}

\renewcommand{\H}{{\mathcal H}}

\newcommand{\operp}{$\bigcirc$\kern-.91em{$\perp$}}

\newcommand{\tr}{\operatorname{Tr}}
\newcommand{\spn}{\operatorname{span}}

\def\be{\begin{eqnarray}}
\def\ee{\end{eqnarray}}
\def\bee{\begin{eqnarray*}}
\def\eee{\end{eqnarray*}}
\def\ot{\otimes}

\begin{document}

\title[Distinguishing Quantum States via One-Way LOCC]{Operator and Graph Theoretic Techniques for Distinguishing Quantum States via One-Way LOCC}
\author[D.W.Kribs, C.Mintah, M.Nathanson, R.Pereira]{David W. Kribs$^{1}$, Comfort Mintah$^{1}$, Michael Nathanson$^2$, Rajesh Pereira$^{1}$}

\address{$^1$Department of Mathematics \& Statistics, University of Guelph, Guelph, ON, Canada N1G 2W1}
\address{$^2$Department of Mathematics, Harvard University, Cambridge, MA, USA 02138}

\begin{abstract}
We bring together in one place some of the main results and applications from our recent works in quantum information theory, in which we have brought techniques from operator theory, operator algebras, and graph theory for the first time to investigate the topic of distinguishability of sets of quantum states in quantum communication, with particular reference to the framework of one-way local quantum operations and classical communication (LOCC). We also derive a new graph-theoretic description of distinguishability in the case of a single qubit sender.
\end{abstract}

\subjclass[2010]{47L25, 47L90, 46B28, 81P15, 81P45, 81R15}

\keywords{quantum communication, quantum states, local operations and classical communication, operator system, operator algebra, quantum error correction, product states, graph theory.}


\maketitle
%

\section{Introduction}

The communication paradigm called local (quantum) operations and classical communication, usually denoted by its acronym LOCC, is fundamental to quantum information theory, and includes many central topics such as quantum teleportation, data hiding, and many of their derivations \cite{Teleportation, terhal2001hiding,eggeling2002hiding}. The somewhat more restricted version called one-way LOCC, in which communicating parties must perform their measurements in a prescribed order, has received expanded attention as being more tractable mathematically while still capturing many of the more important communication scenarios   \cite{Walgate-2000,Nathanson-2005,fan2004distinguishability,N13,cosentino2013small,yu2012four,kribs2017operator, kribsquantum2019,lattice2019,kribs2020vector}. A particularly important subclass of problems in the subject, involves the determination of when sets of known quantum states can be distinguished using only LOCC operations or some subset thereof.

Our work in the theory of LOCC \cite{kribs2017operator,kribsquantum2019,lattice2019,kribs2020vector} has for the first time brought techniques and tools from operator theory, operator algebras, and graph theory to the basic theory of quantum state distinguishability in one-way LOCC. Given the overlapping nature of some of our results and applications, including improvements on some results as our work progressed, we felt a review paper bringing together a selection of main features from our works could be a useful contribution to the literature. In addition to this exposition, we derive a new graph-theoretic description of one-way distinguishability in an important special case, that of a single qubit sender.

This paper is organized as follows. In Section~2 we give necessary preliminaries, including the mathematical description of one-way LOCC in terms of operator relations. Section~3 includes a brief introduction to the relevant operator structures in our analysis, and then we present some of our main results and applications from \cite{kribs2017operator,kribsquantum2019,lattice2019}. We finish in Section~4 by first giving a brief introduction to our necessary notions from graph theory, then we present one of our main results from \cite{kribs2020vector} and some examples, and we derive a new graph-theoretic description for the case of a single qubit sender.

\section{One-Way LOCC and Operator Relations}

We will use the traditional quantum information notation throughout the paper. In particular, we use the Dirac bra-ket notation for vectors, which labels a given fixed orthonormal basis for $\mathbb{C}^d$, with $d \geq 1$ fixed, as $\{ \ket{i} : 0 \leq i \leq d-1 \}$, and the corresponding dual vectors as $\ket{i}^* = \bra{i}$. For $n \geq 1$, $n$-qudit Hilbert space is the tensor product $(\mathbb{C}^d)^{\otimes n}$, which has an orthonormal basis given by $\ket{i_1 i_2 \cdots i_n}:= \ket{i_1} \otimes \ket{i_2} \otimes \ldots \otimes \ket{i_n}$.

We also denote the set of complex $m\times m$ matrices, for a fixed $m\geq 1$, by $M_m (\mathbb{C})$. Given a finite-dimensional Hilbert space $\mathcal H$, we will write $B(\mathcal H)$ for the algebra of bounded (continuous) linear operators on $\mathcal H$, which can be identified with $M_m (\mathbb{C})$ via matrix representations when $\dim \mathcal H = m$. The Pauli operators play an important role in many of our applications, and are given as matrix representations in the single qubit basis $\{ \ket{0}, \ket{1} \}$ for $\mathcal H = \mathbb C^2$ by:
\[
X = \left( \begin{matrix} 0 & 1 \\ 1 & 0 \end{matrix} \right), \quad
Y = \left( \begin{matrix} 0 & -i \\ i & 0 \end{matrix} \right), \quad
Z = \left( \begin{matrix} 1 & 0 \\ 0 & -1 \end{matrix} \right);
\]
where the operators are described by $X\ket{0} = \ket{1}$, $X\ket{1} = \ket{0}$, etc.


The basic set up for the LOCC framework is as follows: multiple parties share a set of quantum states, on which each party can perform local {\it quantum} operations. They can then transmit their results only using {\it classical} information in prescribed directions.

The key problem we have focussed on in our work is to distinguish quantum states amongst a set of known states, where two parties, called Alice ($A$) and Bob ($B$), can perform quantum measurements on their individual subsystems, and then communicate classically. Further, as general LOCC operations are very difficult to characterize mathematically, we have largely restricted ourselves to the case of {\it one-way} LOCC, where the communication is limited to one predetermined direction, generally from $A$ to $B$. This still captures many key examples and settings (though not all).

Hence the bipartite case we consider makes the following assumptions:
\begin{itemize}
\item Two parties $A$, $B$ are separated physically.
\item They control their (finite-dimensional) subsystem Hilbert spaces $\mathcal{H}_{A}$, $\mathcal{H}_{B}$; for simplicity, we often assume $\mathcal{H}_{A}= \mathcal{H}_{B}=\mathbb{C}^d$ for some fixed $d\geq 2$.
\item The state of the composite system $\mathcal{H}_{A} \otimes \mathcal{H}_{B}$  is assumed to be a pure state amongst a known set of  states  $\mathcal{S}  = \lbrace \ket{ \psi_{i}}  \rbrace_{i} \subseteq \mathcal{H}_{A} \otimes \mathcal{H}_{B}$.
\item The goal of $A$ and $B$ is then to identify   the particular  $i$  using only one-way LOCC measurements.
\end{itemize}

The mathematical description of measurements defined by one-way LOCC protocols is given as follows \cite{N13}.

\begin{definition}
A {\bf one-way LOCC measurement}, with $A$ going first, is a set of positive operators $\mathbb{M} = \lbrace A_{k} \otimes B_{k,j} \rbrace_{k,j}$ on $\mathcal{H}_{A} \otimes \mathcal{H}_{B}$ such that
\[
\sum_{k} A_{k} = I_{A} \quad \mathrm{and} \quad \sum_{j} B_{k,j} = I_{B} \quad \forall k.
\]
\end{definition}

Each of the sets $\{A_k\}_k$, $\{B_{k,j}\}_j$ form what is called a {\it positive operator valued measure} (POVM), on $\mathcal H_A$ and $\mathcal H_B$ respectively. If outcome $A_k \otimes B_{k,j}$ is obtained, for any $k$ and a particular $j$, the conclusion is the prepared state was the state identified with the pair $k,j$. Without loss of generality, one can further assume each $A_k$ is a scalar multiple of a (pure) rank one projection.

\begin{example}
As a very simple and illustrative example, consider the following two Bell basis two-qubit states:
\begin{eqnarray*}
\ket{\Phi _{0}} &=& \dfrac{1}{\sqrt{2}} \left( \ket{0}_{A} \ket{0}_{B} + \ket{1}_{A} \ket{1}_{B} \right)
\\ \ket{\Phi _{1}} &=& \dfrac{1}{\sqrt{2}} \left( \ket{0}_{A} \ket{1}_{B} + \ket{1}_{A} \ket{0}_{B} \right)
\end{eqnarray*}
This set is distinguishable, with the following measurement choices:
\begin{itemize}
\item Alice: $A_{1} = \ket{0}\bra{0}$ and $A_{2} = \ket{1} \bra{1}$.

\item  Bob: $\{B_{1,1}, B_{1,2}\} = \{ \ket{0} \bra{0}, \ket{1}\bra{1}\} = \{B_{2,1}, B_{2,2}\}$.
\end{itemize}

If Alice gets outcome $0$, then tells Bob, who after measurement gets outcome $0$, then the state is $ \ket{\psi_{0}}$. Similarly it would be $\ket{\psi_{1}}$ if Bob measured  a $1$.
\end{example}

Notationally, we shall let $\ket{\Phi}$ be the standard {\it maximally entangled state} on two-qudit space $\mathbb{C}^d \otimes\mathbb{C}^d$;
$
\ket{\Phi} = \frac{1}{\sqrt{d}}\big( \ket{00}+\ldots + \ket{d-1 \, d-1}\big).
$
We recall that every maximally entangled state on two-qudit space is then of the form $(I\otimes V)\ket{\Phi}$ where $V$ is unitary on $\mathbb C^d$.

The following result of Nathanson \cite{N13} frames one-way LOCC distinguishability in terms of operator relations and was the starting point for our collaboration.

\begin{theorem}
Let  $\{ U_i \}$ be operators on $\mathbb{C}^d$, and let $\mathcal{S} = \{ \ket{\psi_i}= (I\otimes U_i)\ket{\Phi}  \} \subseteq \mathbb{C}^d \otimes\mathbb{C}^d$ be a set of orthogonal states. Then the following conditions are equivalent:
\begin{itemize}
\item[$(i)$] The elements of $\mathcal{S}$ can be distinguished with one-way LOCC.
\item[$(ii)$] There exists a set of states $\{\ket{\phi_k}\}_{k=1}^r\subseteq \mathbb{C}^d$ and positive numbers $\{ m_k \}$ such that $\sum_{k} m_k \kb{\phi_k}{\phi_k} = I$ and for all $k$ and $i \ne j$,
\begin{equation*}
\bra{\phi_k} U^*_j U_i \ket{\phi_k} = 0.
\end{equation*}

\item[$(iii)$] There is a $d \times r$ partial isometry matrix $W$ such that $WW^* = I_d$, and for all $i \ne j$, every diagonal entry of the $r \times r$ matrix $W^*U_j^*U_i W$ is equal to zero.
\end{itemize}
\end{theorem}

Conceptually, the states $\ket{\phi_k}$ are determined by Alice's (rank one) measurement operators, and the orthogonality of the states $\{ U_i \ket{\phi_k} \}_i$ for every $k$, allows Bob to distinguish $i$.
In the example above, note that $\ket{\Phi _{0}} = (I_2 \otimes I_2) \ket{\Phi}$ and $\ket{\Phi_1} =  (I_2\otimes X)\ket{\Phi}$,
where $X$ is the single-qubit Pauli bit flip operator. Here we have $M_1 = I_2$, $M_2 = X$, and $d=2=r$. So we can take $W = I_2$ and note that $M_j^* M_i = X$ for $i\neq j$.

\section{Operator Structures and One-Way LOCC}

Our initial work \cite{kribs2017operator} identified the importance of certain operator structures for distinguishing various sets of quantum states using one-way LOCC. The following result encompassed our first observation and readily follows from Nathanson's result.  It suggested deeper operator theoretic connections to the mathematics of one-way LOCC lying in the background. 

Let $\Delta : M_d (\mathbb{C}) \rightarrow M_d (\mathbb{C})$ be the `map to diagonal' on $d \times d$ matrices; that is, $\Delta$ zeros out all off-diagonal entries of a matrix but leaves its diagonal entries unchanged, and so there is an orthonormal basis $\{ \ket{k} \}$ for $\mathbb{C}^d$ such that $\Delta(\rho) = \sum_{k=1}^d \kb{k}{k} \rho \kb{k}{k}$ is the (von Neumann) measurement map defined by the basis.

\begin{proposition}  Let $\{ P_k\}_{k=1}^n$ be a set of $d\times d$ permutation matrices and let $\mathcal{S}= \{(I\otimes P_k)\ket{\Phi}\}$ be the set of corresponding maximally entangled states on $\C^d\otimes\C^d$.  Then the following conditions are equivalent:
\begin{enumerate}
\item The states in $\mathcal{S}$ are distinguishable by one-way LOCC.

\item $\Delta(P_j^*P_i)=0$ whenever $i\neq j$.

\end{enumerate}
\end{proposition}


The null space of the map to diagonal operator has a special structure, it is a linear subspace which is closed under taking adjoints.  This observation led us to consider the following notions for the first time in the context of LOCC state distinguishability. First we recall the basic structure theory for finite-dimensional $C^*$-algebras, for instance as exhibited in \cite{davidson1996c,paulsen2002completely}. Every such algebra  is unitarily equivalent to an orthogonal direct sum of ampliated full matrix algebras of the form $\bigoplus_{k}(M_{m_k}(\mathbb{C})\otimes I_{n_k})$ for some positive integers $m_k, n_k \geq 1$. Further, the algebra is unital if it contains the identity operator.

\begin{definition}
Let $\mathfrak{A}$ be a unital $C^*$-algebra.  Any linear subspace  $\mathfrak{S}$ contained in $\mathfrak{A}$ which contains the identity and is closed under taking adjoints is called an {\bf operator system}.
\end{definition}

Within the setting of such operator structures, the following notion turns out to be key for us.

\begin{definition}
Let $\mathcal{H}$ be a Hilbert space and let $\mathfrak{S}\subseteq B(\mathcal{H})$ be a set of operators on $\mathcal{H}$ that form an operator system.  A vector $\ket{\psi}\in \mathcal{H}$ is said to be a {\bf separating vector} for $\mathfrak{S}$ if $A\ket{\psi}\neq 0$ whenever $A$ is a nonzero element of $\mathfrak{S}$; in other words, $A\ket{\psi} = B\ket{\psi}$ with $A,B\in \mathfrak S$ implies $A=B$.
\end{definition}

If $\mathcal{H}$ is finite-dimensional and $\mathfrak{S}$ is closed under multiplication, and hence a $C^*$-algebra, then we may use the decomposition above for such algebras to determine the existence of a separating vector as follows. This result was proved in \cite{pereira2}.

\begin{theorem}
The $C^*$-algebra $\bigoplus_{k} (M_{m_k}(\mathbb{C})\otimes I_{n_k})$ has a separating vector if and only if $n_k \ge m_k$ for all $k$.
\end{theorem}


In the case of the diagonal algebra $\mathfrak{A}_\Delta$, the set of $d \times d$ diagonal matrices (and so $\mathfrak{A}_\Delta \cong \mathbb{C}^d$),  we have $m_k= n_k = 1$ for all $1\leq k \leq d$, and hence $\mathfrak{A}_\Delta$ has a separating vector; an example of which can easily be written down:
$
\ket{\psi} = \frac{1}{\sqrt{d}}( \ket{0} + \ldots + \ket{d-1}).
$

Taken together, these notions and our early results led us to the following general theorem on operator structures and one-way LOCC distinguishability. The first version of the result was proved in \cite{kribs2017operator}, and the refined improvement as stated below was established in \cite{kribsquantum2019}.

\begin{theorem}\label{opthm}
Let $\{U_i\}$ be a set of operators on $\mathbb{C}^d$ and suppose the operator system $\mathfrak{S}_0 = \mathrm{span}\,\{U_i^* U_j, I \}_{i \neq j}$ is closed under multiplication and hence is a C$^*$-algebra. Then $\mathcal S = \{(I\otimes U_i)\ket{\Phi}\}$ is distinguishable by one-way LOCC if and only if  $\mathfrak{S}_0$ has a separating vector.
\end{theorem}

The proof of the theorem starts with the observation that if $\mathfrak S_0$ has a separating vector $\ket{\psi}$, then the states $\{ U_i \ket{\psi} \}$ are linearly independent and Bob can use this fact together with Alice's outcome to distinguish the states.

As a straightforward application of the theorem, consider the following class of states. We recall that a set of matrices $\{ U_k\}$ have a
{\it simultaneous Schmidt decomposition} if there are unitary matrices $V$ and $W$ and complex diagonal matrices $\{ D_k\}$ such that  for each $k$, $U_k=VD_kW$.

\begin{corollary}
Any set of orthonormal states $\{(I\otimes U_i)\ket{\Phi}\}_{i=1}^n$, for which the matrices $U_i$ have a simultaneous Schmidt decomposition, are distinguishable by one-way LOCC.
\end{corollary}

The basic idea of the proof in this case, is to note the operator system structure satisfies $\mathfrak{S}_0 = \mathrm{span}\,\{U_i^* U_j, I \}_{i \neq j}=W^*\mathfrak{A}_\Delta W$ for some unitary $W$.  Since $W^*\mathfrak{A}_\Delta W \cong \mathbb{C}^d$ has a separating vector and hence Theorem \ref{opthm} applies. We note that the operator system $\mathfrak{S}_0$ was studied in a similar context in \cite{duan2013zero}.

\begin{remark}
Towards further applications, including those discussed below, note how the theorem gives a road map to generate sets of indistinguishable states based on these operator structures. Given the decomposition of the algebra generated by an operator system $\mathfrak{A} = \mathrm{Alg}\,(\mathfrak{S}_0) \cong \oplus_k (M_{m_k}(\mathbb{C})\otimes I_{n_k})$, then $\mathfrak{A}$ has a separating vector if and only if $n_k\geq m_k$ for all $k$. Thus, to find instances of sets of {\it indistinguishable states}, we can look for sets $\{U_i\}$ such that $\mathfrak{S}_0 = \mathfrak{A}$ and $m_k > n_k$ for some $k$.
Hence, we are led to consider sets of operators $\{U_i\}$, such that the set is closed under multiplication, taking adjoints, and taking inverses (up to scalar multiples).
\end{remark}

\subsection{Application: States from the Stabilizer Formalism for Quantum Error Correction}

In \cite{kribsquantum2019}, we developed connections between quantum error correction \cite{sho95a,ste96a,bennett1996mixed,gottesmanstab,knilllaflamme,klp05}
and the study of one-way LOCC, including the fact that every one-way LOCC protocol naturally defines a quantum error correcting code defined by the distinguishable states. This also led to new derivations of some known results and new examples of distinguishable states. Here we present one of the applications from that paper.

Sets of unitary operators with the features discussed above are plentiful in one of the foundational areas of quantum error correction, the `stablizer formalism' \cite{gottesmanstab}, which gives a toolbox for generating and identifying codes from the Pauli group.

Let $\mathcal{P}_{n}$ be the $n$-qubit Pauli group; that is, the unitary subgroup on $(\mathbb{C}^2)^{\otimes n}$ with generating set as follows: $$\mathcal{P}_{n} : = \langle  \pm iI; X_{j}, Y_j, Z_{j} : 1 \leq j \leq n \rangle\,,$$
where  $X_1 = X\otimes I\otimes \cdots \otimes I = X \otimes I^{\otimes (n-1)}$, etc.

The {\it Clifford group} is the normalizer subgroup of $\mathcal P_n$ inside the group of $n$-qubit unitary operators. It is known that if $G$ is a subgroup of $\mathcal{P}_n$, with $S_0 = \{ g_1,\ldots g_m\}$ a minimal generating set for a maximal abelian subgroup $S$ of $G$, then there exists a unitary $U$ in the Clifford group such that $U^* g_j U = Z_j$, for all $1\leq j \leq m$. This allows a focus on the generating Pauli operators for deriving more general results.

In \cite{kribsquantum2019} we proved the following. For succinctness we use terminology from the stabilzer formalism in the theorem hypotheses without giving precise details here, as they are not necessary to appreciate the result.

\begin{theorem}
Let $\{ U_i \} \subseteq \mathcal{P}_{n}$ be a complete set of $4^k$ encoded logical Pauli operators for a stabilizer $k$-qubit code on $n$-qubit Hilbert space. Then the set of states $\mathcal S = \{ (I\otimes U_i)\ket{\Phi} \}$ is distinguishable by one-way LOCC if and only if $k \leq \frac{n}{2}$.
\end{theorem}

The basic idea behind this proof as an application of the result above is as follows: The $4^k$ element set $\mathcal{P}_{n,k} =  \langle X_{j}, Z_{j} : 1 \leq j \leq k \rangle / \{ \pm iI\}$, form a set of encoded operations for the code subspace (up to unitary equivalence). But $\mathfrak{S}_0 := \mathrm{span}\,(\mathcal P_{n,k}) = \mathrm{Alg}\,(\mathcal P_{n,k}) = M_{2^k}\otimes I_{2^{n-k}}$. Hence, from the theorem above, the states $\mathcal S$ are distinguishable by one-way LOCC if and only if $\mathfrak{S}_0$ has a separating vector if and only if $2^k \leq 2^{n-k}$, or equivalently $2k \leq n$.

\begin{remark}
The upper bound in this result ($2k=n$) gives sets that saturate a known bound \cite{ghosh2004distinguishability,Nathanson-2005} for the size of one-way distinguishable sets of maximally entangled states on $\mathbb{C}^d\otimes \mathbb{C}^d$ ($d=2^n$). For $2k<n$, this produces (non-trivial) distinguishable sets, which is significant as it is known \cite{cosentino2013small} that many sets defined from $\mathcal P_n$ with less than $2^n$ operators (here $4^k < 2^n$) are not distinguishable even with positive partial transpose operations, and hence not with one-way LOCC.
\end{remark}

\subsection{Application: Sets of Indistinguishable Lattice States}

We have also been able to use the `operator structure road map' outlined above to find sets of lattice states \cite{2013positive,cosentino2013small, yu2012four} that are indistinguishable under one-way LOCC. The following is taken from \cite{lattice2019}.

Recall the two-qubit Bell states $\ket{\Phi_0}, \ket{\Phi_1}$ defined above. The rest of the Bell basis is given by:
\begin{equation*}
\ket{\Phi_2} = \frac{ \ket{01} - \ket{10}}{\sqrt{2}} \qquad \ket{\Phi_3} = \frac{ \ket{00} -\ket{11}}{\sqrt{2}} .
\end{equation*}
These states can be naturally identified with the Pauli matrices by $\ket{\Phi_i} = (I \otimes \sigma_i)\ket{\Phi_0}$ and where we write $I=\sigma_0$, $X=\sigma_1$, $Y=\sigma_2$, $Z=\sigma_3$.

The lattice states are a generalization of the Bell states which are very useful in their own right.

\begin{definition}
For $n \geq 1$, the class of {\bf lattice states} $\mathcal L_n$ are given by $n$-tensors of the Bell states;
\[
{\mathcal L}_n = \{ \ket{\Phi_i} : i \in \{0,1,2,3\} \}^{\otimes n} \subseteq \mathbb{C}^{2^n } \otimes  \mathbb{C}^{2^n} .
\]
\end{definition}

States in ${\mathcal L}_n$ can be identified with elements of the Pauli group $\mathcal{P}_{n} = \{ \otimes_{k = 1}^n \sigma_{i_k} \}$, using an extension of the Bell state identification above.

\begin{theorem}\label{LOCCLatticeExample}
For every $n > 1$ and $d = 2^n$, there exist sets of $m$ lattice states in $\mathbb{C}^d\otimes \mathbb{C}^d$ that are not distinguishable with one-way LOCC, where
\begin{eqnarray*}
m = \left\{ \begin{array}{ll} 2\sqrt{2d}-1 & \mbox{ if  $n$  is odd} \cr 3\sqrt{d}-1 & \mbox{ if  $n$  is even.} \end{array} \right.
\end{eqnarray*}
\end{theorem}

\begin{remark}
As discussed further in \cite{lattice2019}, this result is new and can be extended to so-called `generalized Pauli states', where a new proof is given of an established result, and which leads to an improvement for a studied subclass of states \cite{wangetal2016}. The following example illustrates the approach.
\end{remark}

\begin{example}
For an example in $\mathcal L_n$, with $n$ fixed, we can set
\begin{eqnarray*}
S_1 &=&\{ I^{\otimes i} \otimes Z \otimes I^{\otimes n-i-1}\}_{i = 0}^{k-1}  \\ S_2 &=&\{ I^{\otimes i} \otimes Z \otimes I^{\otimes n-i-1}\}_{i = k}^{n-1}  \cup \{ X^{\otimes n}\}.
\end{eqnarray*}
It is easy to check that the algebra generated by $S_1$ has dimension $2^k$; the algebra generated by $S_2$ has dimension $2^{n+1-k}$; and the algebra generated by $S_1 \cup S_2$ has dimension $2^{n+1}$. This gives us:
\begin{eqnarray*}
S = \left( \{ I ,Z \}^{\otimes k}\otimes I^{\otimes (n-k)} \right) &\cup & \left( I^{\otimes k} \otimes \{ I ,Z \}^{\otimes n-k}  \right) \\ &\cup& \left( { X}^{\otimes k} \otimes \{ X,Y \}^{\otimes n-k}  \right)
\end{eqnarray*}
with $|S| = 2^k + 2^{n-k+1} - 1$, which achieves its minimum when $k = \lfloor \frac{n}{2}+1 \rfloor$ and $|S| \in  \{2\sqrt{2d}-1, 3\sqrt{d}-1\}$.
\end{example}

\section{Graph Theory and Distinguishing Product States}

The following section contains a brief review of the main results and some applications from \cite{kribs2020vector}, in which we used graph theory to study the problem of distinguishing sets of product states via one-way LOCC. Our graph-theoretic work in LOCC is ongoing, and here we give a new graph-theoretic perspective and proof for the case that Alice only has access to a single qubit Hilbert space.

We shall write $G = (V,E)$ for a {\it simple graph} with vertex set $V$ and edge set $E$. For $v,w\in V$, we write $v \sim w$ if the edge $\{ v,w \}\in E$.  The {\it complement} of $G$ is the graph $\overline{G} = (V, \overline{E})$, where the edge set $\overline{E}$ consists of all two-element sets from $V$ that are not in $E$. Another graph $G'$ is a {\it subgraph} of $G$, written $G' \leq G$, if $V' \subseteq V$ and $E' \subseteq E$ with $v,w\in V'$ whenever $\{ v,w\} \in E'$.

Given a graph $G= (V,E)$, a function $\phi: V \rightarrow \mathbb{C}^d\backslash \{0\} $ is an {\it orthogonal representation} of $G$ if for all vertices $v_i \ne v_j\in V$,
\[
v_i \not\sim v_j \iff \langle \phi(v_i), \phi(v_j) \rangle = 0 .
\]
Orthogonal representations have been discussed in graph theory, for instance \cite{fallat2007minimum,lovasz1989orthogonal}. Note the biconditional in the definition, which is stronger than conditions for graph colouring. This allows us to uniquely define the graph associated with a function $\phi$.

We introduced a graph-theoretic perspective to distinguishing product states as follows.
Suppose we are given a set of product states $\{ \ket{\psi^A_k}\otimes \ket{\psi^B_k} \}_{k=1}^r$ on $\mathcal H_A \otimes \mathcal H_B$. The graph of these states from Alice's perspective is the unique graph $G_A$ with vertex set $V = \{ 1,2, \ldots, r\}$ such that the map $k \mapsto  \ket{\psi^A_k}$ is an orthogonal representation of $G_A$. Likewise, the graph of the states from Bob's perspective is the graph $G_B$ with vertex set $V$ such that $k \mapsto \ket{\psi^B_k}$ is an orthogonal representation of $G_B$.
Observe that by construction, the set of product states are mutually orthogonal precisely when Alice's graph is a subgraph of the complement of Bob's graph; that is, $G_A \leq \overline{G_B}$.

The following concepts from graph theory are central for us.
Given a graph $G = (V,E)$, a set of graphs $\{ G_i = (V_i, E_i) \}$ {\it covers} $G$ if $V = \cup_i V_i$ and $E = \cup_i E_i$.
A collection of graphs $\{G_i \}$ is a {\it clique cover} for $G$ if $\{G_i \}$ covers $G$ and if each of the $G_i$ is a complete graph (i.e., a clique).
The clique cover number $\mathrm{cc}(G)$ is the smallest possible number of subgraphs contained in a clique cover of $G$.
A clique cover can be thought of as a collection of (not necessarily disjoint) induced subgraphs of $G$, each of which is a complete graph (there is an edge between every pair of vertices) with the condition that every edge is contained in at least one of the cliques.

The following is one of our main results from \cite{kribs2020vector}, and gives a characterization of when product states are one-way LOCC distinguishable in terms of the underlying Alice and Bob graph structures and a related decomposition of Alice's Hilbert space. We note this is a corrected version of the theorem from \cite{kribs2020vector}.  The revision was made to  condition (3), as the earlier version gave a condition that was sufficient but not necessary.

\begin{theorem} \label{thm: one-way LOCC graphs} Given a set of product states in $\H_A \ot \H_B$, let $G_A$ and $G_B$ be the graphs of the states from Alice and Bob's perspectives, respectively.  Let $\phi: V_A \rightarrow \H_A$  be the association of vertices with  Alice's states and assume that the set $\{\phi(v): v\in V\}$ spans $\H_A$.

Then the states are distinguishable with one-way LOCC with Alice measuring first if and only if there exists
\begin{itemize}
\item[(1)] a graph $G$ satisfying $G_A \le G \le \overline{G_B}$,
\item[(2)] a clique cover $\{V_j \}_{j = 1}^k$ of $G$, and,
\item[(3)] a POVM $\{Q_j\}$ on $\mathcal H_A$ such that for all $v\in V_A$, $Q_j\phi(v) \ne 0$ implies that $v \in V_j$.
\end{itemize}
\end{theorem}

The focus on clique decompositions gives tools for building optimal POVMs. We include an example showing a POVM that is not a von Neumann measurement (i.e., the operators are not mutually orthogonal projections).

\begin{example} Let Alice's (unnormalized) states be given in $\C^3$ by
\bee
\ket{\phi_{0}} = \ket{0} +\ket{1}  &\qquad& \ket{\phi_{1}} = \ket{0} + \ket{2} \\
\ket{\phi_{2}} = \ket{0} -\ket{1}  &\qquad& \ket{\phi_{3}} = \ket{0} - \ket{2}  ,
\eee
and Bob's states given so that $G_A = \overline{G_B} = C_4$, the 4-cycle graph.  The clique cover number of a 4-cycle is 4, which is bigger than our dimension, but we are still able to distinguish the states using the following POVM. We define
\bee
\ket{\psi_{0}} = \ket{0} +\ket{1}+\ket{2}  &\qquad& \ket{\psi_{1}} =-\ket{0} +\ket{1}+\ket{2}  \\
\ket{\psi_{2}} =\ket{0} -\ket{1}+\ket{2} &\qquad& \ket{\psi_{3}} = \ket{0} +\ket{1}-\ket{2}.  \eee
For each $j$, we can then define $Q_j = \frac{1}{4} \kb{\psi_{j}}{\psi_{j}} $.    It is easy to check that $\sum_j Q_j = I$ and that each $Q_j$ picks out an edge of $C_4$ as in the theorem conditions, so we can apply the theorem to show that the states are one-way LOCC distinguishable.
\end{example}

We point the interested reader to \cite{kribs2020vector} for some consequences of this result and related results. For the rest of this section, we will focus on the case of a `low rank' sender, in which this theorem can be entirely stated in graph-theoretic terms.

\subsection{Single Qubit Sender and Graph Theory}

A basic result in LOCC theory  \cite{bennett1999unextendible,divincenzo2003unextendible,halder2020distinguishability} shows that any set of orthogonal product states in $\mathbb{C}^{2} \otimes  \mathbb{C}^{d}$ for arbitrary $d\geq 2$ can be distinguished  via full (two-way) LOCC. This is readily seen to not be the case for one-way LOCC. As a simple example, consider the two-qubit set $\{\ket{00},\ket{10}, \ket{+1}, \ket{-1}\}$ where $\ket{+}=\frac{1}{\sqrt{2}}(\ket{0}+\ket{1})$ and $\ket{-}=\frac{1}{\sqrt{2}}(\ket{0}-\ket{1})$. This set of four two-qubit states cannot be distinguished by one-way LOCC with Alice going first; they can however be distinguished by one-way LOCC with Bob going first and hence also by two-way LOCC.


Nevertheless, one can characterize when one-way distinguishability is possible in the single qubit sender case. As proved in \cite{Walgate-2002} (Theorem~1), the states must take a particularly nice form, in terms of orthogonality on Alice's side versus corresponding orthogonalities on Bob's side. That said, one could argue that the condition from \cite{Walgate-2002} is perhaps not so operationally simple to apply to easily identify when states are distinguishable, if the set is large for example. Here we show how, in the case of product states, Theorem~\ref{thm: one-way LOCC graphs} can be refined for a single qubit sender to yield an entire graph-theoretic set of testable conditions for one-way distinguishability.

\begin{theorem}\label{singlequbit}
A set of orthonormal product states in $\mathbb{C}^{2} \otimes  \mathbb{C}^{d}$, for $d \geq 2$, is distinguishable via one-way LOCC with Alice going first if and only if there is some graph between the two graphs $G_A$ and $\overline{G_B}$ with clique cover number at most two; that is, there is a graph $G$ such that
\begin{equation}\label{singlequbitgraphcond}
G_A \leq G \leq \overline{G_B} \quad \mathit{and} \quad \mathrm{cc}\,(G) \leq 2.
\end{equation}
\end{theorem}

\begin{proof}
For the forward direction, our previous theorem gives the existence of $G$ with a clique cover and POVM $\{Q_j\}$ such that $\phi(v_1)^*\phi(v_2) =0$ implies $\phi(v_1)^* Q_j\phi(v_2) = 0$ for all $j$. This implies that either $\mathrm{cc}(G) = 1$ or else each $Q_j$ is diagonal in the $\{ \phi(v_1), \phi(v_2) \}$ basis. It follows that $\mathrm{cc}(G) \le 2$.

For the backward direction, assume a graph $G$ exists that satisfies the conditions of Eq.~(\ref{singlequbitgraphcond}). We shall consider the two cases as determined by the clique cover number of $G$.

Firstly, if $\mathrm{cc}\,(G) =1$, then $G$ is a complete graph. Since $|G_A| = |G_B|$, $G$ must contain all the vertices of $G_B$, hence $\overline{G_B}$ is a complete graph and all the vertices of $G$ are pairwise disconnected in the graph $G_B$. Thus, it follows that all of Bob's states are pairwise orthogonal in $\mathcal H_B = \mathbb{C}^d$ (whether or not Alice's states are orthogonal). Hence the full set of product states on $\mathcal H_A \otimes \mathcal H_B = \mathbb{C}^{2} \otimes  \mathbb{C}^{d}$ are one-way distinguishable, simply by a measurement that Bob can perform on his states, independent of what Alice does or communicates to Bob.

Now suppose that $\mathrm{cc}\,(G) =2$. Let $\{ G_1, G_2\}$ be a clique cover of $G$, and let $V_0 = V(G_1) \cap V(G_2)$ be the set of vertices in $G$ that are connected to one (and hence every) vertex in each of $G_1$ and $G_2$. Note that $V_0$ is a proper subset of $V(G_i)$, for $i=1,2$, as otherwise $G_1$ would be a subgraph of $G_2$, or vice-versa, and this would incorrectly imply that $\mathrm{cc}\,(G) =1$.

If $V_0$ is empty, then $G_1$ and $G_2$ are disconnected and each is a subgraph of $\overline{G_B}$. As such, Bob's states corresponding to vertices in $G_1$ are mutually orthogonal, and the same is true of his states corresponding to vertices in $G_2$. Moreover, we can define orthogonal (one-dimensional) subspaces of $\mathcal H_A$ by $\mathcal H_i = \mathrm{span}\, \{ \phi_A(v) : v\in V(G_i)\}$, for $i=1,2$. The states are thus distinguishable via one-way LOCC with Alice first performing a measurement defined by the Hilbert space decomposition $\mathcal H_A = \mathcal H_1 \oplus \mathcal H_2$, then sending the outcome ($j=$ 1 or 2) to Bob, who then performs a measurement defined by the orthogonal states $\{ \phi_B(v) : v\in V(G_j)\}$ to determine the state.

If $V_0$ is non-empty, let $v_1 \in V(G_1)\setminus V_0$, and choose $v_2 \in V(G_2)$ such that $v_1$ and $v_2$ have no edge in $G$ (and hence also in $G_A$) connecting them. Note that such a vertex exists in $V(G_2)$ as otherwise $G=G_1\cup G_2$ would be a single clique and so $\mathrm{cc}\,(G)=1$. Also necessarily $v_2\in V(G_2)\setminus V_0$, as $v_1$ connects with all vertices in $V_0$. It follows that $\ket{\psi_i} = \phi_A(v_i)$, $i=1,2$, are orthogonal and hence form an orthonormal basis for $\mathcal H_A = \mathbb{C}^2$. Furthermore, $\bk{\psi_i}{\phi_A(w_i)}\neq 0$ for all $w_i \in V(G_i)$ and $i=1,2$. This sets up a one-way LOCC protocol as follows: Alice measures in the basis $\{ \ket{\psi_1}, \ket{\psi_2}\}$ for $\mathcal H_A$, and communicates the outcome to Bob. As $G\leq \overline{G_B}$ and $G_j$ is a clique, the states $\{ \phi_B(v) : v\in V(G_j)\}$ are mutually orthogonal in $\mathcal H_B$, and so Bob can determine the state by performing a measurement defined by these states and the projection onto the orthogonal complement of their span. This completes the proof.
\end{proof}

It is fairly straightforward to give examples that satisfy $G_A = \overline{G_B}$ and $\mathrm{cc}(G_A)\leq 2$. For instance, the set $\{ \ket{00}, \ket{01}, \ket{1+}, \ket{1-}\}$, where $\ket{+},\ket{-}$ in this case is any qubit basis different than the standard basis. Here, Alice would measure in the standard basis, and then Bob would measure in the basis suggested by Alice's outcome communicated to him. A simple example of a distinguishable set for which $G_A$ is a proper subgraph of $\overline{G_B}$ is given by the standard two-qubit basis $\{ \ket{00}, \ket{01}, \ket{10}, \ket{11}\}$ (left as an easy exercise: $\mathrm{cc}(G_A)=2<4 = \mathrm{cc}(\overline{G_B})$).

We finish by presenting a `nice' indistinguishable example, in that $G_A = \overline{G_B}$, but nevertheless the states fail to be distinguishable due to the failure of the clique cover condition.
\begin{example}
Consider the (unnormalized) states $\ket{\psi_{i}} \in \mathbb{C}^{2} \otimes \mathbb{C}^{3} $, for $1\leq i \leq 5$, defined as follows:
\begin{align*}
\ket{\psi_{1}}& = \ket{1} \otimes \ket{1}\\
\ket{\psi_{2}} &=  \ket{0} \otimes \left(   \ket{0} + \ket{1}  \right)\\
\ket{\psi_{3}} &=  \ket{0} \otimes \left(   \ket{0} - \ket{1}  \right)\\
\ket{\psi_{4}} &=  \left(   \ket{0} + \ket{1}  \right)  \otimes \ket{2} \\
\ket{\psi_{5}} &=  \left(   \ket{0} - \ket{1}  \right)  \otimes \ket{2} \\
\end{align*}

These states {\it can} be distinguished with full LOCC operations, with Bob measuring first followed by Alice and then once more by Bob. We show that Bob's initial measurement is necessary and that one-way distinguishability is impossible with Alice going first.

Observe here that $G_A = \overline{G}_{B}$. Moreover, the complement of Bob's graph has clique cover number
$
\mathrm{cc}(\overline{G_{B})} = 4 > 2,
$
with a clique cover of minimal size given by the vertex sets:
\begin{eqnarray*}
\Big\{ \{1,5\} , \,\, \{ 1,4\}, \,\, \{2,3,4\}, \,\,
\{2,3,5\} \Big\}
\end{eqnarray*}
Hence the set of states is not distinguishable via one-way LOCC with Alice going first.
\end{example}

\section{Conclusion}
One of the appeals of quantum information theory is how it builds on expertise in a wide range of areas in physics, mathematics, computer science, and engineering.  This review paper highlights several aspects of the fruitful interplay between operator theory and questions of local quantum state distinguishability; and the previous section adds graph theory into the mix. One-way LOCC is a simply-constructed problem with real physical implications, and our work continues to develop effective tools to study it.

\strut

{\noindent}{\it Acknowledgements.} D.W.K. was partly supported by NSERC. C.M. was partly supported by Mitacs and the African Institute for Mathematical Sciences. R.P. was partly supported by NSERC. M.N. acknowledges the support of Saint Mary's College of California, where most of this work was completed.

\bibliographystyle{plain}

\bibliography{KMNP_Bibfile}

\end{document}